%% file: main.tex
\documentclass[letterpaper, 10 pt, conference]{ieeeconf} 
\IEEEoverridecommandlockouts                           
\usepackage{amsmath} 
\usepackage{amssymb}
\usepackage{xcolor}
\usepackage{algorithm}
\usepackage{algcompatible}
\usepackage{graphicx}

\algnewcommand{\algorithmicand}{\textbf{and}}
\algnewcommand{\algorithmicor}{\textbf{or}}
\algnewcommand{\OR}{\algorithmicor}
\algnewcommand{\AND}{\algorithmicand}
\algrenewcommand\algorithmicindent{1.0em}

\newtheorem{theorem}{Theorem}
\newtheorem{assumption}{Assumption}
\newtheorem{remark}{Remark}
\newtheorem{definition}{Definition}

\newcommand{\col}{\mathrm{col}}
\newcommand{\diag}{\mathrm{diag}}

\title{\LARGE \bf
Hierarchical Cyber-Attack Detection in Large-Scale Interconnected Systems
}

\author{Twan Keijzer$^{*}$, Alexander J. Gallo$^{*}$ and Riccardo M.G. Ferrari$^{*}$
\thanks{$^{*}$Delft Centre for Systems and Control,
        Delft University of Technology, The Netherlands.
        {\tt\small \{t.keijzer,a.j.gallo,r.ferrari\}@tudelft.nl}}
\thanks{This paper has been partially supported by the AIMWIND project, which is funded by the Research Council of Norway under grant no. 312486.}
}

\begin{document}

\maketitle
\thispagestyle{empty}
\pagestyle{empty}

\begin{abstract}
    In this paper we present a hierarchical scheme to detect cyber-attacks in a hierarchical control architecture for large-scale interconnected systems (LSS). We consider the LSS as a network of physically coupled subsystems, equipped with a two-layer controller: on the local level, decentralized controllers guarantee overall stability and reference tracking; on the supervisory level, a centralized coordinator sets references for the local regulators. We present a scheme to detect attacks that occur at the local level, with malicious agents capable of affecting the local control. The detection scheme is computed at the supervisory level, requiring only limited exchange of data and model knowledge. We offer detailed theoretical analysis of the proposed scheme, highlighting its detection properties in terms of robustness, detectability and stealthiness conditions.
\end{abstract}

\section{Introduction}\label{sec:intro}
\input{Sections/introduction.tex}

\section{Problem Statement}\label{sec:probstatement}
\input{Sections/problemstatement.tex}

\section{Hierarchical Control}\label{sec:hictrl}
\input{Sections/control.tex}

\section{Hierarchical Attack Detection}\label{sec:distest}
\input{Sections/disturbance.tex}

\section{Robustness and Detectability Analysis}\label{sec:detection}
\input{Sections/detection.tex}

\section{Simulation Example}\label{sec:simulation}
\input{Sections/simulation.tex}

\section{Conclusion}\label{sec:conclusion}
\input{Sections/conclusion.tex}

\bibliographystyle{IEEEtran}
\bibliography{biblio}

\end{document}

%% file: Sections/introduction.tex
\noindent Modern engineering systems, ranging from large-scale infrastructure like electrical grids, water distribution and traffic networks, as well as industrial plants and consumer goods, have an increasing penetration of distributed computational resources, and a heavy reliance on communication networks. This improves performance and efficiency of these systems, and has led to the definition of \textit{cyber-physical systems} (CPS)
\cite{baheti2011cyber} as an analytical framework.
Although the integration of these ``cyber'' resources has great benefits, it also leads to the exposure to malicious tampering, as has been made evident in recent years by some high profile cases of cyber-attacks \cite{lee2008cyber,falliere2011w32}.
Because many of these systems are safety-critical \cite{giraldo2017security}, methods have been developed over the past decade to detect, isolate and mitigate attacks in CPS \cite{sandberg2015cyberphysical,chen2020guest}.

CPS are often also large-scale systems (LSS) \cite{lunze1992feedback,siljak2011decentralized}, i.e., they require a large number of states to be described, and are spatially distributed over large areas.
As such, centralized control architectures, with a single regulator managing the inputs for the entire system, are not feasible.
Thus, non-centralized control architectures have been developed, which rely on partitioning the LSS into \textit{subsystems} \cite{lunze1992feedback}, each of which is physically interconnected with neighbouring subsystems.
These control architectures can be further classified as decentralized, distributed, and hierarchical, all of which have attracted extensive literature (see for example the surveys \cite{scattolini2009architectures,chanfreut2021survey}, and references therein).

Non-centralized monitoring architectures have been proposed for both fault diagnosis \cite{blanke2016distributed,teixeira2014,arrichiello2015observer,davoodi2016simultaneous,boem2017distributed,boem2020distributed,khalili2020distributed} and cyber-attack detection \cite{Deghat2019a,Keijzer2021DetectionOC,anguluri2019centralized,barboni2020detection,gallo2020distributed}, predominantly on distributed and decentralized architectures. Furthermore, there are however application based papers, such as \cite{salinas2016} addressing energy theft, where implicitly a hierarchical framework is used for cyber-attack detection.

On the other hand, here we focus on hierarchical architectures.
In hierarchical control, the computational advantages of distributed or decentralized controllers are blended with the coordination capability of a centralized, \textit{supervisory}, layer \cite{scattolini2009architectures}.
Hierarchical control indeed appears naturally in large-scale interconnected systems, as it is an architecture that allows for multiple degrees of complexity and coordination to be integrated.

In this paper, we propose a hierarchical cyber-attack detection scheme which leverages the physical coupling between local subsystems to detect attacks. By computing two estimates, at the local and supervisory level, sufficient redundancy is introduced to perform diagnoses. For this detection scheme:
\begin{enumerate}
    \item cyber-attacks fully compromising one or more local controllers can be detected at the supervisory level;
    \item the supervisory level requires only a reduced order representation of the subsystem dynamics for detection, allowing for reduced computational overhead;
    \item a thorough theoretical analysis of the detection properties is presented, including robustness, guaranteed attack detectability conditions, and existence conditions for locally and globally stealthy attacks.
\end{enumerate}
The use of physical coupling to generate the necessary redundancy for detection has been proven beneficial in distributed cyber-attack diagnosis schemes \cite{barboni2020detection,gallo2020distributed}, and indeed we show that it is a critical aspect of our proposed hierarchical diagnoser.
Here, subsystem model knowledge and measurements are used to define a local estimate of the physical coupling, which is then compared to a supervisory estimate computed from global knowledge.

We make use of a set-based detection scheme, which has a rich history as a detection method in the FDI literature, as can be seen in, e.g., \cite{boem2020distributed,NIKOUKHAH19981345,SCOTT20141580,SCOTT2016126,puig2010fault,rostampour2020privatized}, where this list does not have the pretence of being exhaustive. In this paper, but without loss of generality, we adopt constrained zonotopes \cite{SCOTT2016126},
which are proven to offer numerical advantages with respect to other set representations.

The rest of the paper is structured as follows: in Section~\ref{sec:probstatement} we introduce the problem formulation, describing the dynamics of the large-scale interconnected systems, and formalizing the attacks considered in this paper. Then, in Section~\ref{sec:hictrl}, we outline a hierarchical controller. In Section~\ref{sec:distest}, we present our proposed hierarchical scheme for estimating the physical interconnection between subsystems, and give the definition of our detection test. Following this, in Section~\ref{sec:detection}, we offer theoretical analysis for our proposed method. Finally, in Section~\ref{sec:simulation} we provide some numerical results, and in Section~\ref{sec:conclusion} we offer concluding remarks.

\paragraph*{Notation}
\noindent For a matrix $A$,  $\|A\|$ represents the induced Euclidian norm of $A$, and $\ker A$ its right null-space. $\col_{i\in\mathcal{I}}[x_i]$ and $\diag_{i\in\mathcal{I}}[x_i]$ denote respectively the column and block-diagonal concatenation of vectors or matrices $x_i~\forall~i\in\mathcal{I} \triangleq \{1,\dots,N\}$. Given matrices $A_{ij}, i,j \in \mathcal I$ of appropriate dimensions, $A = [A_{ij}]$ denotes the block matrix with $A_{ij}$ in the $i,j$-th block. For sets $\mathcal{A}\subseteq \mathbb{R}^n$ and $\mathcal{B} \subseteq \mathbb{R}^n$ we denote with $\oplus$ the Minkowski sum, $\mathcal A\oplus \mathcal B \triangleq \{x\in \mathbb{R}^n : x = a + b, \,\forall a \in \mathcal A, b\in \mathcal B\}$; with $\mathcal{A}\ominus\mathcal{B}$ the erosion, or Pontryagin difference, of $\mathcal{A}$ by $\mathcal{B}$, $A\ominus B\triangleq \{\zeta \in \mathbb{R}^n | \zeta + b \in A, \, \forall b \in \mathcal B\}$, and therefore $\mathcal{A}\ominus\mathcal{A}=0$ \cite{blanchini2008set}. The cartesian product of two sets $\mathcal A$ and $\mathcal B$ is defined as $\mathcal A \times \mathcal B$.
Additionally, for a set $\mathcal A$,  $\text{Vol}(\mathcal{A})$ denotes its volume and $\mathcal{A}'$ its set complement. Furthermore, constrained zonotope $\mathcal{Z}\subseteq \mathbb{R}^n$ is defined as ${\mathcal{Z}\triangleq\left\{\zeta\in\mathbb{R}^n:\zeta=c+G\beta,\,\|\beta\|_\infty\leq 1, A\beta=b\right\}}$. Efficient definitions of set operations with constrained zonotopes can be found in \cite{SCOTT2016126}. 
Finally, with $\bigotimes_{i\in\mathcal{I}} A_i$ we intend $A_1 \times A_2 \times \dots \times A_N$.

%% file: Sections/problemstatement.tex
\noindent 
We consider a linear time-invariant \emph{Large-Scale System} (LSS) which is partitioned into $N$ physically coupled subsystems $\mathcal{S}_i, i \in \{1,\dots,N\} \triangleq \mathcal{N}$.
The dynamics of each subsystem is written as
\begin{equation}\label{eq:LSS}
    \begin{cases}
        x_i(k+1) = A_{ii} x_i(k) + B_i u_i(k) + d_i(k)\,,\\
        d_i(k) = \sum_{j\in\mathcal{N}_i} A_{ij}x_j(k)
    \end{cases}
\end{equation}
where $x_i \in \mathbb{R}^{n_i}, u_i\in\mathbb{R}^{m_i}$ are the subsystem state and control input.
The term $d_i(k)$ accounts for the physical coupling between subsystems, where $\mathcal{N}_i \triangleq \{j\in\{1,2,\dots,N\}: \partial {x}_i/\partial x_j \neq 0\} \subseteq \mathcal N$ is the set of \textit{neighbors} of $\mathcal S_i$, i.e., those subsystems which physically influence the dynamics of $\mathcal S_i$.
\begin{assumption}
For all $i \in \mathcal N$, $(A_{ii},B_i)$ is controllable.
$\hfill\triangleleft$
\end{assumption}
We suppose the LSS is regulated via a hierarchical control architecture, composed of two layers, as shown in Figure~\ref{fig:HierarchLSS}, with the following characteristics, detailed in Section~\ref{sec:hictrl}:
\begin{itemize}
    \item[--] locally, each subsystem is regulated by a decentralized controller $\mathcal{C}_i^\ell$.
    This guarantees LSS stability and is capable of tracking a suitably defined reference $r_i$;
    \item[--] at the supervisory level, a controller $\mathcal{C}^s$ is designed to provide appropriate references $r_i$ to the local controllers, thus providing coordination for the LSS.
\end{itemize}

\begin{figure}[b]
    \centering
    \includegraphics[width=0.8\linewidth,trim=1cm 2.5cm 0cm 4cm]{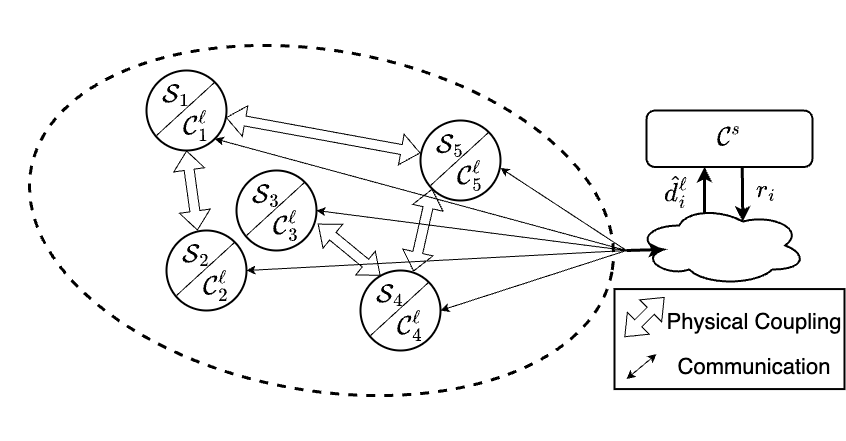}
    \caption{The considered hierarchical control and detection architecture.}
    \label{fig:HierarchLSS}
\end{figure}
\subsection{Cyber-attack vulnerability}
\noindent In this work, we consider a cyber-attack carried out by an agent capable of fully compromising a subset of the subsystems and their controllers. In order to clearly define the problem we address, we introduce the following assumption.
\begin{assumption}\label{ass:initAtk}
    An attacker may attack subsystems with indexes $\mathcal{I}_a \subseteq \mathcal N$ from some time $k_a \geq 0$.
    $\hfill\triangleleft$
\end{assumption}
The considered attack is modelled as $u_i = u_i^0 + a^u_i$ where $u_i^0$ is the healthy input as obtained by $\mathcal{C}^\ell_i$ and $a^u_{i}$ is, without loss of generality, an additive attack. 

The control architecture considered in this paper requires a communication network that links the supervisory controller $\mathcal C^s$ to all the local controllers $\mathcal C^\ell_i$. Thus, $\mathcal C^s$ represents a single point of failure, and if it were compromised by a malicious agent, it could steer the entire LSS to any desired operating condition. Given this premise, we suppose that the hardware and software of the supervisory controller are suitably designed to give a higher degree of protection, and therefore $\mathcal C^s$ cannot be subject to attacks.

\subsection{Hierarchical Cyber-Attack Detection}
\noindent
Let us, before introducing the hierarchical control architecture considered in this paper, briefly give an overview of our proposed method.
The detection architecture relies on comparing two estimates of $d_i$: one of the estimates is computed locally for each subsystem in $\mathcal N$, based on local model knowledge and measurements; a second estimate is computed at the supervisory level, using information on the references followed by each subsystem. During operation, the local estimate is transmitted to the supervisory level, where a detection test is performed.
Detection is then performed by comparing two sets bounding the nominal coupling $d_i$ given the local and supervisory estimates and their respective estimation errors. The proposed method allows these sets to be constructed at the supervisory level in a computationally efficiently way using only limited model knowledge of the LSS.  

%% file: Sections/control.tex
\noindent 
Let us now describe the hierarchical controllers regulating the LSS. We stress again that this control architecture requires a communication network to exchange information between each local controller $\mathcal C_i^\ell$ and the supervisory controller $\mathcal C^s$, whilst not communicating amongst each other.
The control architecture is represented in Figure~\ref{fig:HierarchLSS}.

\subsection{Local controllers $\mathcal C_i^\ell$}\label{sec:hictrl:loc}
\noindent 
We consider a decentralized $\mathcal{C}_i^\ell$ designed to ensure stability of the LSS, while locally tracking a reference. The reference $r_i \in \mathbb{R}^{q_i}$ is set by the supervisory controller $\mathcal C^s$, and follows the dynamics:
\begin{equation}\label{eq:regDyn}
    r_i(k+1) = S_i r_i(k),
\end{equation}
with $S_i$ such that its eigenvalues have modulus no smaller than one \cite{francis1977linear}. Thus, defining the output tracking error $e_i \in \mathbb{R}^{p_i}$, $e_i = C_i x_i + Q_i r_i$, the controller's tracking objective is to define $u_i$ such that nominally, i.e., when the system is not under attack, $\lim_{k\rightarrow\infty} e_i(k) = 0$. Specifically, we focus on $\mathcal C_i^\ell$ capable of solving the full-information regulator problem \cite{francis1977linear}.

\begin{assumption}\label{ass:reg}
    For all $i \in \mathcal N$, the full state $x_i$ is measured by the controller $\mathcal C_i^\ell$ for regulation purposes.
    $\hfill\triangleleft$
\end{assumption}

\begin{remark}
    Assumption~\ref{ass:reg}, although potentially limiting, is introduced here to simplify the analysis of the controllers $C^\ell_i$, which is not the primary focus of this paper.
    $\hfill\triangleleft$
\end{remark}
\begin{assumption}\label{ass:kerC}
    For each subsystem $\mathcal S_i, i \in \mathcal N$, $\mathrm{ker} C_i \subseteq \mathrm{ker} A_{ji}$, for all $\mathcal S_j$ such that $i \in \mathcal N_j$.\
    $\hfill\triangleleft$
\end{assumption}
\noindent We consider the control law
\begin{equation}\label{eq:ctrl}
    \mathcal C_i^\ell : \quad
        u_i^0(k) = K_i x_i(k) + L_i r_i(k)\,,
\end{equation}
where $K_i$ and $L_i$ satisfy the output regulation problem \cite{francis1977linear}.
Specifically, $K_i$ is designed such that, while $r_i = 0, \forall i \in \mathcal N$, the closed-loop LSS dynamics is asymptotically stable (for methods to design $K_i$, see \cite{lunze1992feedback,siljak2011decentralized} and references therein). Before defining $L_i$, we introduce the following.
\begin{assumption}\label{ass:reg2}
    For each subsystem $i \in \mathcal N$, the matrices $A_i, B_i, C_i, Q_i, S_i$ are such that there exist $\Pi_i \in \mathbb{R}^{n_i\times q_i}$ and $\Gamma_i \in \mathbb{R}^{m_i\times q_i}$ such that:
    \begin{subequations}\label{eq:regEq}
        \begin{align}
            \Pi_iS_i  &= A_i \Pi_i + B_i \Gamma_i\label{eq:regEq:1}\\
            0 &= C_i \Pi_i + Q_i\label{eq:regEq:2}
        \end{align}
    \end{subequations}
    holds.
    $\hfill\triangleleft$
\end{assumption}
Assumption~\ref{ass:reg2} guarantees that the so-called regulator equations \eqref{eq:regEq} can be solved, and thus 
$L_i$ can be defined as
\begin{equation}
    L_i = \Gamma_i - K_i \Pi_i\,,
\end{equation}
which guarantees that $x_i \rightarrow \Pi_i r_i$ as $k\rightarrow \infty$ \cite{francis1977linear}, supposing $d_i = 0, \forall i \in \mathcal N$. Note that satisfaction of \eqref{eq:regEq:2} guarantees that $\epsilon_i\triangleq x_i - \Pi_i r_i = 0$ implies $C_i\epsilon_i=e_i = 0$.
On the other hand, for $d_i \neq 0$, we have that the dynamics of the state tracking error $\epsilon_i(k) = x_i(k) - \Pi_i r_i(k)$ are
\begin{equation}\label{eq:reg:track}
    \begin{split}
        &\epsilon_i(k+1) 
        = A_{ii} x_i(k) + B_i u_i^0(k) + d_i(k) - \Pi_iS_i r_i(k)\\
        &\overset{\eqref{eq:ctrl},\eqref{eq:regEq:1}}{=} A_{ii}^{cl} (x_i(k)-\Pi_i r_i(k))  + d_i(k)\,,
    \end{split}
\end{equation}
where $A_{ii}^{cl}\triangleq A_{ii}+B_iK_i$. Given the stability of $A_{ii}^{cl}$, by design of $K_i$, $\epsilon_i$ is bounded for bounded $d_i$. 

\subsection{Supervisory-level controller $\mathcal{C}^h$}
\noindent Having presented the local decentralized tracking controller $\mathcal{C}_i^\ell$, we can briefly discuss the design of $\mathcal{C}^s$.
As previously stated, the objective of $\mathcal{C}^s$ is to design $r_i, \forall i \in \mathcal{N}$ such that some level of coordination between subsystems is possible.
Although the specific design of $r_i$ is dependent on the type of application considered, and is out of the scope of this paper, we introduce some basic characteristics that must be included in its design.
We suppose that $\mathcal C^s$ defines $r_i$ as a piecewise constant signal. This in turn implies that the reference dynamics in \eqref{eq:regDyn} are $r_i(k+1) = r_i(k)$ for almost all $k$, and therefore that $S_i = I_{p_i}, \forall i \in \mathcal N$; furthermore, set  $Q_i = -I_{p_i}, \forall i \in \mathcal N$.
This definition of $S_i$ still allows for the references to be changed at discrete time instances. However, depending on the rate of convergence of $A_{ii}^{cl}$, it is important to specify a minimum time between switching times, and a maximum step in $r_i$ \cite{BARCELLI2011277}.

%% file: Sections/disturbance.tex
\noindent 
The hierarchical cyber-attack detection scheme presented in this paper uses the physical interconnection between subsystems to perform detection at the supervisory level. To this end two estimates of this physical interconnection are computed. The so-called \emph{local estimate} depends on local measurements as well as local model information, and is calculated at each local subsystem and communicated to the supervisory level. The so-called \emph{supervisory estimate} is calculated at the supervisory level and depends on knowledge of the interconnection and a simplified model of the local dynamics. These two estimates, along with their estimation uncertainties, are compared at the supervisory level to detect cyber-attacks using a set-based approach. 

\subsection{Supervisory Estimate}
\noindent 
The supervisory estimate of the physical interaction between all subsystems is computed as
\begin{equation}\label{eq:high_est}
    \hat{d}^s= A_c \Pi r\,,
\end{equation}
where $A_c = A - \diag_{i\in\mathcal N}[A_{ii}]$, $A = [A_{ij}]$, $\Pi = \diag_{i\in\mathcal N} [\Pi_i]$ and $r = \col_{i\in\mathcal N} [r_i]$.
Thus, the estimation error is \\
$e^s\triangleq d-\hat{d}^s = A_c \epsilon$, where $d = \col_{i\in\mathcal N} [d_i]$ and $\epsilon=\col_{i\in\mathcal N} [\epsilon_i]$. Thus, in nominal conditions the error can be bounded as $e^s \in \mathcal E^s$, with:
\begin{equation}\label{eq:sup_err_set}
    \mathcal E^s(k) = A_c \mathcal E(k),
\end{equation}
where $\mathcal E = \bigotimes_{i\in\mathcal N} \mathcal E_i$, and 
\begin{equation}\label{eq:normset}
    \mathcal E_i(k) \triangleq \{\zeta \in \mathbb{R}^{n} | \|\zeta\|\leq \bar{\epsilon}_i(k)\}\,,
\end{equation}
where 
$\bar{\epsilon}_i(k)$ guarantees $\|\epsilon_i(k)\| \leq \bar{\epsilon}_i(k), \forall k < k_a$, such that $\epsilon_i(k) \in \mathcal E_i(k), \forall k < k_a$: the trajectory of $\bar{\epsilon}_i$ is the result of the dynamics
\begin{equation}\label{eq:norm_update}
    \begin{split}
        \bar{\epsilon}_i(k+1) = b_i \bar{\epsilon}_i(k) + \|\hat{d}_i^s(k)\| + \|\mathcal E_i^s(k)\| \\
        + \|\Pi_i\|\|r_i(k+1) - r_i(k)\|\,,
    \end{split}
\end{equation}
which are defined bounding \eqref{eq:reg:track} via the triangle inequality, where $b_i \in [0,1]$ is defined such that $\|{A_{ii}^{cl}}^k\epsilon_i(0)\| \leq b_i^k \bar{\epsilon}_i^0$ for all $i \in \mathcal N$, $\bar{\epsilon}_i^0$ is an appropriately defined initial bound, $\hat{d}_i^s(k)$ are the components of $\hat{d}^s(k)$ relating to $\mathcal S_i$, and $\mathcal E_i^s$ is the projection of $\mathcal E$ onto the space relating to $\mathcal S_i$.
Furthermore, $\|r_i(k+1) - r_i(k)\|$ is added to bound the effect of reference changes.
By using a bound on the norm of $\epsilon_i$, only the rate of convergence $b_i$, is needed at the supervisory level. Then, via \eqref{eq:high_est}, \eqref{eq:sup_err_set} the \emph{supervisory estimation set} is defined
\begin{equation}\label{eq:sup_est_set}
    \mathcal{D}^s\triangleq\hat{d}^s\oplus\mathcal{E}^s\,.
\end{equation}

Thus, $d(k) \in \mathcal D^s$ holds by construction for all $k < k_a$.
Note that $\mathcal{D}^s$ only depends on $r$ and is therefore not affected by the considered cyber-attacks. Furthermore $\mathcal E_i$ is represented as a hyper-sphere in $n_i$-dimensions. As such sets are not closed under matrix multiplication as done in \eqref{eq:sup_err_set} we define a constrained zonotope $\bar{\mathcal E}_i$ that encloses $\mathcal{E}_i$ and use this for further calculation. An exact representation of $\mathcal{E}_i$ as a constrained zonotope requires a number of generators approaching infinity which is computationally infeasible, while a hyper-cube with $n_i$ generators may not be sufficiently accurate. Therefore, the representation of $\bar{\mathcal E}_i$ should be constructed to balance computational requirements with accuracy.
\begin{remark}
    Computing enclosing sets can be computationally intensive, however $\mathcal E_i$ is always a hyper-sphere centered in $0$.
    Thus, an enclosing set can be computed once off-line for the unit hyper-sphere, to only be re-scaled on-line.
    $\hfill\triangleleft$
\end{remark}

\subsection{Local Estimate}
\noindent 
The local estimator in each local subsystem is based on the reduced unknown input observer (R-UIO) by \cite{Lan2016}. Here, we exploit this R-UIO for the estimation of the physical coupling between local subsystems, by defining an extended system
\begin{equation}\label{eq:UIO_system}
\left\{\begin{aligned}
    \bar{x}_i(k+1) &= \bar{A}_{ii}\bar{x}_i(k) + \bar{B}_iu_i(k) +\bar{d}_i(k)\\
    y_i(k)&=\bar{C}_i\bar{x}_i(k)
    \end{aligned}\right.\\
\end{equation}
where $\bar{x}_i = \left[\begin{matrix}
    x_i\\d_i
    \end{matrix}\right]$, $\bar{d}_i(k)=\left[\begin{matrix}0\\d_i(k+1)-d_i(k)
    \end{matrix}\right]$,\\
$\bar{A}_{ii} = \left[\begin{matrix}
    A_{ii} & I\\ 0 & I
    \end{matrix}\right]$, $
    \bar{C}_i=\left[\begin{matrix}
    I&0
    \end{matrix}\right]$, $\bar{B}_i=\left[\begin{matrix}
   B_i\\0\end{matrix}\right]$.
Then, the R-UIO takes the form
\begin{equation}\label{eq:UIO}
\left\{\begin{aligned}
    \xi_i(k+1)&=M_i\xi_i(k) + G_iu_i(k) + R_iy_i(k)\,,\\
    \hat{d}_i^{\ell}(k) &= \xi_i(k) +H_iy_i(k) \,,
    \end{aligned}\right.
\end{equation}
where $\xi_i$ is the observer state, $\hat{d}_i^{\ell}$ is the local disturbance estimate, and $M_i$, $G_i$, $R_i$, and $H_i$ are designed such that
\begin{equation}\label{eq:RUIO:mat}
\begin{aligned}
    \left|\left|M_i\right|\right|<1\,,\\
    M_iT_i-R_i\bar{C}_i-T_i\bar{A}_i=0\,,\\
    T_i+H_i\bar{C}_i-L_i=0\,,\\
    G_i-T_i\bar{B}_i=0\,.
\end{aligned}
\end{equation}
An efficient design approach can be found in \cite{Lan2016}. 
Following definition of the matrices in \eqref{eq:RUIO:mat}, the dynamics of the local estimation error $e_i^\ell = \hat{d}_i^\ell - d_i$ can be written as
\begin{equation}\label{eq:est_err_low}
    e_i^{\ell}(k+1) = M_i e_i^\ell(k) - T_i \bar{d}_i(k)\,,
\end{equation}
which, given \eqref{eq:RUIO:mat}, converges to a neighborhood of the origin asymptotically. Let us define $e^{\ell}=\col_{i\in\mathcal{N}}[e^{\ell}_i]$ such that the stability of \eqref{eq:est_err_low} implies that $e^{\ell} \in \mathcal E^\ell$, where
\begin{equation}\label{eq:UIO_err_set}
    \mathcal{E}^{\ell}(k+1) = M \mathcal{E}^\ell(k)\oplus(- T)
    \bar{\mathcal{D}}(k)\,,
\end{equation}
where $M=\diag_{i\in\mathcal{N}}[M_i]$, $T=\diag_{i\in\mathcal{N}}[T_i]$ and $\bar{\mathcal{D}}(k)$ is defined as $\bar{\mathcal{D}}(k)=\mathcal{D}^s(k)\oplus(-{\mathcal{D}^s}(k+1))$. Note that \eqref{eq:UIO_err_set} is evaluated at the supervisory level to obtain $\mathcal{E}^{\ell}$. Thus, because $\mathcal{D}^s$ is known, it can be used to obtain $\bar{\mathcal{D}}$. Furthermore, it can be seen the supervisory level requires no knowledge of the local dynamics to obtain $\mathcal{E}^{\ell}$, although it requires knowledge of the R-UIO parameters $M$ and $T$. 

Note that $\hat{d}_i^\ell$ is calculated at the local level using \eqref{eq:UIO}, while detection is performed at the supervisory level. 
As such, it must be transmitted from the local to the supervisory level.
Therefore, given that the cyber-attack can fully compromise a local controller, the estimate $\hat{d}_i^\ell$ transmitted to the supervisory level might also be subject to a cyber-attack. 
To allow for this additional attack, we define the local estimate sent to the supervisory level as $\hat{d}_i^{\ell,a} = \hat{d}_i^\ell + a_i^d$. Based on \eqref{eq:UIO} and \eqref{eq:UIO_err_set}, and the additional cyber-attack we define the \emph{local estimation set} at the supervisory level as
 \begin{equation}\label{eq:loc_est_set}
     \mathcal{D}^\ell=\hat{d}^{\ell,a}\oplus\mathcal{E}^\ell\,,
 \end{equation}
 where $\hat d^{\ell,a}=\mbox{col}_{i\in\mathcal{N}}\hat d^{\ell,a}_i$.
\subsection{Cyber-Attack Detection Condition}
\noindent At the supervisory level cyber-attack detection will be performed using sets $\mathcal D^\ell$ and $\mathcal D^s$ as previously defined. By construction, for $k < k_a$, the physical coupling satisfies
\begin{equation}\label{eq:intersect}
    d \in \mathcal D^\ell \cap \mathcal D^s\triangleq \mathcal{D}\,.
\end{equation}
As $d$ itself is not known, the condition $d\notin\mathcal{D}$ cannot directly be used for cyber-attack detection. Alternatively, we can check the detection condition 
\begin{equation}\label{eq:detect_cond}
    \mathcal{D}=\emptyset\,,
\end{equation}
which implies $d\notin\mathcal{D}$. A summary of the scheme is given in Algorithms~\ref{alg:low} and~\ref{alg:high}.
\begin{remark}
All sets described in this section are defined as constrained zonotopes, which are closed under matrix multiplication, addition and intersection, thus justifying the equalities in \eqref{eq:sup_err_set}, \eqref{eq:UIO_err_set}, and \eqref{eq:intersect}.
$\hfill\triangleleft$
\end{remark}
\begin{algorithm}[ht]
\caption{Actions at Each Local Subsystem $i$}\label{alg:low}
\textbf{Initialize:} Determine $\hat{d}_i^\ell(0)$.
\begin{algorithmic}
\FORALL{$k\geq0$}
\STATE Compute $u_i(k)$ \eqref{eq:ctrl} and  $\hat{d}_i^\ell(k)$ \eqref{eq:UIO}.
\STATE Transmit $\hat{d}_i^{\ell,a}$ to the supervisory level.
\ENDFOR
\end{algorithmic}
\end{algorithm}\vspace{-0.5cm}
\begin{algorithm}[ht]
\caption{Actions at the Supervisory Level}\label{alg:high}
\textbf{Initialize:} Determine $\mathcal{E}^{s}(0)$ and $\mathcal{E}^\ell(0)$.
\begin{algorithmic}
\FORALL{$k\geq 0$}
\FORALL{$j \in \mathcal N$}
    \STATE Determine $r_j$,  
    \STATE Compute $\mathcal{E}_j$ (\ref{eq:normset}, \ref{eq:norm_update}).
    \STATE Enclose $\mathcal{E}_j$ by $\bar{\mathcal{E}}_j$.
\ENDFOR
    \STATE Receive $\hat{d}^{\ell,a}$ from local subsystems.
    \STATE Compute $\hat{d}^s$ \eqref{eq:high_est}, $\mathcal{D}^s$ (\ref{eq:sup_err_set}, \ref{eq:sup_est_set}),  $\mathcal{D}^{\ell}$ (\ref{eq:UIO_err_set}, \ref{eq:loc_est_set}), $\mathcal{D}$ \eqref{eq:intersect}.
    \STATE Check the Detection Condition \eqref{eq:detect_cond}.
\ENDFOR
\end{algorithmic}
\end{algorithm}

%% file: Sections/detection.tex
\noindent
The properties of the proposed cyber-attack detection scheme are analysed based on changes in the local disturbance estimate $\hat{d}_i^{\ell,a}$, as it is the only part of the detection algorithm affected by the considered cyber-attacks. 
Let us introduce $\hat{d}_i^{\ell,0}$ and $\hat{d}^{\ell,0}=\mbox{col}_{i\in\mathcal{N}}\hat{d}_i^{\ell,0}$ as the nominal local estimate of the physical coupling, i.e. $\hat{d}^{\ell,0}(k)=\hat{d}^{\ell,a}(k)$ if $a_i^u(k) = a_i^d(k) = 0~\forall~k,i\in\mathcal{N}$. Accordingly, we define $\mathcal{D}^{\ell,0}\triangleq \hat{d}^{\ell,0}\oplus\mathcal{E}^\ell$.
Similarly, we define $x_i^a$ as the part of the state driven by an attack $a_i^u$, satisfying dynamics: $x_i^a(k+1) = A_{ii}^{cl}x_i^a(k) + B_i a_i^u(k)$. Because of superposition in linear systems, $x_i$ can be written as $x_i = x_i^0 + x_i^a$, where $x_i^0$ is the nominal state.

\begin{theorem}\label{thm:robustness}
It can be guaranteed no false detection occurs, i.e. $\mathcal{D}^s\cap \mathcal{D}^\ell\neq\emptyset$ for all $k\leq k_a$.
$\hfill\square$
\end{theorem}
\begin{proof}
The sets $\mathcal{D}^s$ and $\mathcal{D}^\ell$ are such that, in nominal conditions, $d\in\mathcal{D}^{s}$ and $d\in\mathcal{D}^\ell$ hold by construction.
Therefore, $d\in\mathcal{D}^s\cap \mathcal{D}^\ell\neq\emptyset$ for all $k\leq k_a$
\end{proof}
\begin{theorem}\label{thm:detectability}
Any attack for which $\hat{d}^{\ell,a}-\hat{d}^{\ell,0}\in \mathcal{D}^{s'}\ominus\mathcal{D}^{\ell,0}$ is guaranteed to be detected.
$\hfill\square$
\end{theorem}
\begin{proof}
Define $\mathcal{A}=\mathcal{D}^{s'}\ominus\mathcal{D}^{\ell,0}$, such that $\mathcal{D}^{\ell,0}\oplus\mathcal{A}=\mathcal{D}^{s'}$, which implies $(\mathcal{D}^{\ell,0} \oplus \mathcal{A})\cap\mathcal{D}^{s}=\emptyset$ \cite{blanchini2008set}. This proves that an attack is detected if $\hat{d}^{\ell,a}-\hat{d}^{\ell,0}\in\mathcal{A}$.
\end{proof}
\noindent To analyze existence conditions for stealthy attacks, let us introduce the following definitions.
\begin{definition}[Locally Stealthy]
A cyber-attack $\{a_i^u\neq0,a_i^d\neq0\}$ is locally stealthy if $\hat{d}_i^{\ell,a}(k)=\hat{d}_i^{\ell,0}(k)~\forall k\geq 0$.
$\hfill\triangleleft$
\end{definition}
\begin{definition}[Globally Stealthy]
A cyber-attack $\{a_i^u\neq0,a_i^d\neq0\}$ is globally stealthy if it is locally stealthy and $\hat{d}_j^{\ell,a}(k)=\hat{d}_j^{\ell,0}(k)~\forall k,j~\text{s.t.}~i\in\mathcal{N}_j$.
$\hfill\triangleleft$
\end{definition}
\begin{theorem}\label{thm:loc_stealthy}
For all $i \in \mathcal I_a$, there exists $a^u_i\neq0$ and $a^d_i\neq0$ for which $\hat{d}_i^{\ell,a} = \hat{d}_i^{\ell,0}$, i.e. 
the attack is locally stealthy.
$\hfill\square$
\end{theorem}
\begin{proof}
For sake of space, the proofs of Theorem~\ref{thm:loc_stealthy} and Theorem~\ref{thm:glob_stealthy} are omitted.
\end{proof}

\begin{theorem}\label{thm:glob_stealthy}
There exists a globally stealthy attack if and only if $a_i^u$ and $a_i^d$ satisfy Theorem~\ref{thm:loc_stealthy}, and $a_i^u$ is such that $x_i^a \in \mathrm{ker}A_{ji}$, for all $j \in \{l\in\mathcal N: i \in \mathcal N_l\}\triangleq \widehat{\mathcal N}_i$, $i \in \mathcal I_a$.
$\hfill\square$
\end{theorem}

\begin{remark}
Although the construction of locally and globally stealthy attacks is not given, we note that to perform such attacks, malicious agents require information relating to more than only the attacked subsystems. Indeed, to be locally stealthy, $a_i^d$ must be designed to consider the effect that $x_i^a$ has on the neighbors of $i \in \mathcal I_a$ through the physical coupling. This is different to other definitions of locally stealthy attacks available in literature \cite{barboni2020detection}, where information about the local subsystem dynamics is sufficient to perform such attacks.
$\hfill\triangleleft$
\end{remark}

%% file: Sections/simulation.tex
\noindent To demonstrate the effectiveness of the proposed hierarchical detection scheme, we apply it to a system with four subsystems that are physically connected in series. The system is discretized with a time-step of $0.25~[s]$ and modelled by \eqref{eq:LSS} and \eqref{eq:ctrl} using the following parameters $\forall i\in \mathcal{N}$
\begin{equation*}
\begin{aligned}
    \allowdisplaybreaks
    &A_{ii} =\left[\begin{matrix} 1&0.25\\-0.055 & 0.995\end{matrix}\right];B_i=\left[\begin{matrix} 0 \\ 0.05\end{matrix}\right] ~;~ K_i = \left[\begin{matrix} -0.284 \\ -2.100\end{matrix}\right]^\top\\
    &L_i = 1.484~;~A_{ij} =\left[\begin{matrix} 0.005&0\\0 & 0\end{matrix}\right] \forall j\in \{i-1,i+1\}\cap \mathcal{N}
 \end{aligned}
\end{equation*}
Cyber-attack detection is performed using Algorithms \ref{alg:low} and \ref{alg:high}. Here the parameters for R-UIO \eqref{eq:UIO} are found using \cite{Lan2016}. Furthermore, $b_i$ in \eqref{eq:norm_update} is chosen as $0.955$.
\begin{figure}[t]
    \centering
    {\includegraphics[width=\linewidth,trim={1.8cm 10.2cm 2.4cm 11.2cm},clip]{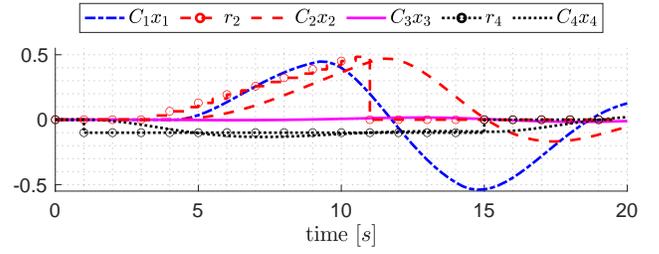}}
    \caption{Reference and Tracking performance of all subsystems.}
    \label{fig:refout}
\end{figure}
\begin{figure}[t]
    \centering
    \includegraphics[width=\linewidth,trim={1.2cm 11.3cm 1.2cm 11.4cm},clip]{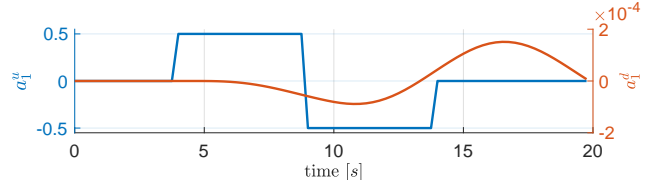}
    \caption{The implemented attack, where $a^d_1$ makes the attack locally stealthy.}
    \label{fig:attack}
\end{figure}
The references $r_i$ that are non-zero are shown in Figure \ref{fig:refout}. The system is corrupted by an attacker in subsystem $1$ as shown in Figure~\ref{fig:attack}. Here $a^u_i$ disturbs the system and $a^d_i$ is designed to make the attack locally stealthy (see Theorem~\ref{thm:loc_stealthy}).

Figure \ref{fig:refout} also shows the reference tracking performance of all subsystems. One can see that all systems show reasonable tracking except the attacked subsystem $1$. The detection of attack $a^u_1$ is shown as red shade in Figure \ref{fig:detection}. Additionally, a metric is introduced which gives more insight into how close a detection is. This metric is denoted $i_f$ and is defined as
\begin{equation*}
    i_f = \frac{\text{Vol}(\mathcal{D})}{\min(\text{Vol}(\mathcal{D}^s),\text{Vol}(\mathcal{D}^\ell))}\,,
\end{equation*}
such that $i_f=1$ if set $\mathcal{D}^\ell$ encloses set $\mathcal{D}^s$ or vice versa, and $i_f=0$ if the sets do not intersect.

Figure \ref{fig:detection} shows an estimate of $i_f$ obtained using a Monte Carlo method using 1000 random samples of $\mathcal{D}^\ell$ and $\mathcal{D}^s$. This estimate is used as obtaining the exact volume of a constrained zonotope is computationally hard. One can see that $i_f$ is normally $1$ and becomes smaller than $1$ around times of detection of attack $a^u_1$ and momentarily around $1$ and $3~[s]$ due to excitation of the system. The temporary lack of detection at $12-14~[s]$ is caused by the change of sign of the attack at $8 [s]$, which causes the system to temporarily be in a condition that resembles nominal behaviour.

Figure \ref{fig:setplot} shows a projection of the sets $\mathcal{D}^\ell$ and $\mathcal{D}^s$ on the axes representing $d_2$ and $d_4$. Due to the attack local estimation set $\mathcal{D}^\ell$ moves while the supervisory estimation set $\mathcal{D}^s$ remains unaffected. The projections shown in Figure \ref{fig:setplot} are used here to illustrate the detection method, and any use of these projections for identification has not been studied.

%% file: Sections/conclusion.tex
\noindent 
Hierarchical control architectures are commonly used in industrial control systems, however very little work exists on fault or cyber-attack detection schemes with the same architecture. In this paper a hierarchical cyber-attack detection scheme has been presented which utilizes two estimates of the physical coupling between local subsystem. These are compared at the supervisory level to detect cyber-attacks on the local subsystems. The redundancy of information contained in these two estimates can be used for detection. Indeed, the local estimate uses local model knowledge and measurements to estimate how each subsystem is affected by the coupling; on the other hand, the supervisory estimate uses knowledge of the reference tracked by each subsystem to estimate how each subsystem affects its neighbours. A set based method using constrained zonotope representation is presented to calculate the estimation errors, which are compared for cyber-attack detection. It is proven that the detection method is robust, and existence conditions for guaranteed detectability, and locally and globally stealthy attacks are given. In future work, we intend to extend the proposed detection scheme to multi-rate systems.
\begin{figure}[t]
    \centering
    \includegraphics[width=\linewidth,trim={1.8cm 11.3cm 2.4cm 11.3cm},clip]{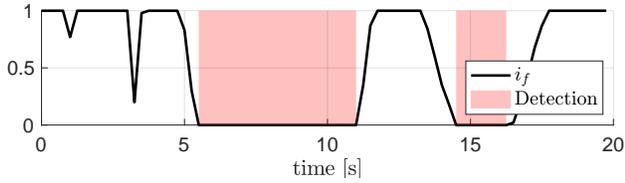}
    \caption{Detection times and fraction of sets $\mathcal{D}^h$ and $\mathcal{D}^\ell$ that are intersecting.}
    \label{fig:detection}
\end{figure}
\begin{figure}[t]
    \centering
    \includegraphics[width=\linewidth,trim={2cm 11cm 2cm 10.9cm},clip]{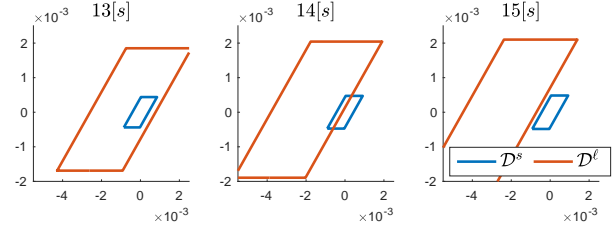}
    \caption{Visualization of the set based attack detection using a projection to show the estimate sets for $\mathcal S_i, i \in \{2,4\}$.}
    \label{fig:setplot}
\end{figure}